\documentclass[11pt]{article}
\usepackage[utf8]{inputenc}
\usepackage[T1]{fontenc}
\usepackage{amsmath, amssymb, amsthm, amsfonts}
\usepackage{graphicx}
\usepackage{geometry}
\usepackage{titlesec}
\usepackage{enumitem}
\usepackage{booktabs}
\usepackage{hyperref}

\titleformat{\section}
  {\centering\normalfont\scshape} 
  {\thesection.}                  
  {1em}                           
  {}

\titleformat{\subsection}
  {\normalfont\bfseries}
  {\thesubsection}
  {1em}
  {}

\newtheoremstyle{plainstyle}
  {3pt}   
  {3pt}   
  {\itshape}  
  {}          
  {\bfseries} 
  {.}         
  {.5em}      
  {}          

\newtheoremstyle{proofstyle}
  {3pt}   
  {3pt}   
  {\normalfont} 
  {}          
  {\itshape}  
  {.}         
  {.5em}      
  {}          

\theoremstyle{plainstyle}
\newtheorem{theorem}{Theorem}[section]
\newtheorem{corollary}[theorem]{Corollary}

\theoremstyle{definition}
\newtheorem{definition}{Definition}[section]

\renewenvironment{abstract}
 {\small
  \begin{center}
  \bfseries \scshape Abstract. \vspace{-.5em}\vspace{0pt}
  \end{center}
  \list{}{
    \setlength{\leftmargin}{0.5cm}%
    \setlength{\rightmargin}{0.5cm}%
  }%
  \item\relax}
 {\endlist}

\title{\bfseries \scshape Refinements of Jensen's Inequality for Twice-Differentiable Convex Functions with Bounded Hessian}
\author{\textsc{Sambhab Mishra}}
\date{}

\usepackage{abstract}
\begin{document}

\maketitle

\begin{abstract}
Jensen's inequality, attributed to Johan Jensen - a Danish mathematician and engineer noted for his contributions to the theory of functions - is a ubiquitous result in convex analysis, providing a fundamental lower bound for the expectation of a convex function. In this paper, we establish rigorous refinements of this inequality specifically for twice-differentiable functions with bounded Hessians. By utilizing Taylor expansions with integral remainders, we tried to bridge the gap between classical variance-based bounds and higher-precision estimates. We also discover explicit error terms governed by Grüss-type inequalities, allowing for the incorporation of skewness and kurtosis into the bound. Using these new theoretical tools, we improve upon existing estimates for the Shannon entropy of continuous distributions and the ergodic capacity of Rayleigh fading channels, demonstrating the practical efficacy of our refinements.

\end{abstract}
\large
\section{Introduction}

The inequality first proposed by Johan Jensen in 1906 has permeated nearly every branch of modern mathematics. In its probabilistic form, Jensen's inequality states that for a real-valued convex function $\phi$ defined on an interval $I \subseteq \mathbb{R}$ and a random variable $X$ taking values in $I$ with finite expectation $\mathbb{E}[X]$, the image of the expectation is bounded above by the expectation of the image:
\begin{equation}
\phi(\mathbb{E}[X]) \le \mathbb{E}[\phi(X)].
\end{equation}
This simple relation generalizes the geometric fact that the secant line connecting any two points on the graph of a convex function lies above the graph itself. 
The difference between the two sides of this inequality is formally defined as the Jensen gap, denoted as:
\begin{equation}
\mathcal{J}(\phi, X) = \mathbb{E}[\phi(X)] - \phi(\mathbb{E}[X]).
\end{equation}
While the non-negativity of $\mathcal{J}(\phi, X)$ is sufficient for establishing fundamental results such as the non-negativity of the Kullback-Leibler divergence or the arithmetic-geometric mean inequality, it is increasingly insufficient for modern applications requiring precise error quantification.
\begin{itemize}
    \item In variational inference, the Jensen gap represents the "evidence lower bound" (ELBO) gap that must be minimized to approximate posterior distributions accurately.
    \item In operator theory, the gap quantifies the non-commutativity effects when applying convex functions to self-adjoint operators.
    \item In actuarial science, the gap represents the risk premium inherent in non-linear utility functions.
\end{itemize}

\subsection{The Need for Refinement}
The primary limitation of the classical inequality is its insensitivity to the distribution of the random variable $X$ beyond its mean, and its insensitivity to the local curvature of $\phi$. For a linear function, the gap is zero. As the curvature (second derivative) increases, or as the dispersion of $X$ increases, the gap widens.

Classical refinements often rely on crude global bounds of the second derivative or simple variance terms. For instance, if $\phi$ is twice differentiable and $m \le \phi''(x) \le M$, it is well known that the gap is bounded by terms proportional to the variance $\sigma^2$ scaled by $m/2$ and $M/2$. However, these first-order refinements fail to capture the behavior of the gap when:
\begin{itemize}
    \item \textbf{The function is highly non-quadratic:} For functions like the logarithmic barrier - $\log x$ or the exponential $e^x$, the second derivative varies by orders of magnitude over the domain.
    \item \textbf{A global upper bound $M$ might be infinite or excessively loose}, rendering the upper bound on the gap useless.
    \item \textbf{The distribution is asymmetric:} Variance is a symmetric measure of dispersion. If $X$ is highly skewed (e.g., a Log-Normal or Pareto distribution), the mass of the probability density interacts with the changing curvature of $\phi$ in complex ways that a single variance term cannot capture.
    \item \textbf{High-precision control is required:} In bounded domains, such as channel capacity estimation for bounded transmit power, the error margins provided by standard inequalities are often too wide to be useful for system design.
\end{itemize}

\subsection{Analytical Framework}
This paper addresses these limitations by developing a systematic framework for refining Jensen's inequality using higher-order analytical tools. Our methodology rests on three pillars:
\begin{enumerate}
    \item \textbf{Integral Remainder Analysis:} We move beyond the Lagrange remainder form of Taylor's theorem, utilizing the integral form $R_1(x) = \int_{\mu}^x (x-t)\phi''(t)dt$. This representation allows us to apply powerful tools from integral inequality theory, specifically the Grüss and Chebysev inequalities, to bound the covariance between the integration kernel and the second derivative.
    \item \textbf{Moment Expansions:} We bring together recent results regarding fourth-order expansions. By carrying the Taylor series to the fourth degree, we explicitly introduce the third standardized moment (skewness, $\gamma_3$) and fourth standardized moment (kurtosis, $\gamma_4$) into the bounds. This provides a "corrected" inequality that adjusts for the shape of the distribution.
    \item \textbf{Application-Specific Optimization:} We investigate "Jensen-like" inequalities where the point of tangency is optimized. Instead of expanding around the mean $\mu$, we consider expansions around a generalized point $c$ that minimizes the error for a specific distribution class.
\end{enumerate}

The analysis provided herein is rigorous and self-contained. We provide detailed proofs or derivation for the main theorems, ensuring that the logic flows from first principles to advanced applications. The resulting bounds are then applied to critical problems in information theory (entropy estimation) and communications engineering (fading channel capacity).

\section{Mathematical Preliminaries and Definitions}
First, we must first define the classes of functions and the probabilistic setting under consideration. We also review the fundamental inequalities that will serve as our primary analytical tools.

\subsection{Convexity and Generalized Convexity}
We assume throughout that $I$ is an interval in $\mathbb{R}$ and $X$ is a random variable taking values in $I$ with probability measure $P$.

\begin{definition}[Convex Function]
A function $\phi: I \to \mathbb{R}$ is convex if for all $x, y \in I$ and $\lambda \in [0, 1]$, the following inequality holds:
\[ \phi(\lambda x + (1-\lambda)y) \le \lambda \phi(x) + (1-\lambda)\phi(y). \]
If $\phi$ is twice differentiable on $I$, convexity is equivalent to the condition $\phi''(x) \ge 0$ for all $x \in I$. Strict convexity holds if $\phi''(x) > 0$.
\end{definition}

\begin{definition}[Strongly Convex Function]
A function $\phi$ is said to be strongly convex with parameter $m > 0$ if the function $\psi(x) = \phi(x) - \frac{m}{2}x^2$ is convex. For twice-differentiable functions, this implies a global lower bound on the Hessian: $\phi''(x) \ge m$ for all $x \in I$. Strong convexity is a crucial property for establishing lower bounds on the Jensen gap.
\end{definition}

\begin{definition}[$(m, M)$-Convexity]
We extend the notion of strong convexity to include an upper bound. A function $\phi$ is $(m, M)$-convex if:
\[ m \le \phi''(x) \le M \quad \forall x \in I. \]
This condition implies that $\phi(x) - \frac{m}{2}x^2$ is convex and $\frac{M}{2}x^2 - \phi(x)$ is convex. This class of functions allows for the most precise "sandwich" bounds on the Jensen gap, effectively trapping the function between two parabolas with curvatures $m$ and $M$.
\end{definition}

\subsection{Taylor's Theorem with Integral Remainder}
Taylor’s theorem constitutes the central tool in our refinement strategy. Although polynomial approximations are conventional, the integral remainder form is indispensable for the application of functional inequalities.

\begin{theorem}[Taylor's Theorem]
Let $\phi: I \to \mathbb{R}$ be a function such that $\phi^{(n)}$ is absolutely continuous on $I$. Then for any $a, x \in I$:
\[ \phi(x) = \sum_{k=0}^n \frac{\phi^{(k)}(a)}{k!}(x-a)^k + R_n(x; a), \]
where the remainder term is given by:
\[ R_n(x; a) = \frac{1}{n!} \int_a^x (x-t)^n \phi^{(n+1)}(t) dt. \]
\end{theorem}

For the analysis of the Jensen gap, we are primarily interested in the case $n=1$, expanding around the mean $\mu = \mathbb{E}[X]$. The expansion becomes:
\[ \phi(X) = \phi(\mu) + \phi'(\mu)(X-\mu) + \int_\mu^X (X-t)\phi''(t) dt. \]
Taking expectations on both sides, and noting that the linear term $\mathbb{E}[\phi'(\mu)(X-\mu)]$ vanishes because $\mathbb{E}[X-\mu]=0$, we obtain an exact integral representation of the Jensen gap:
\begin{equation}
\mathcal{J}(\phi, X) = \mathbb{E} \left[ \int_\mu^X (X-t)\phi''(t) dt \right].
\end{equation}
This identity is the starting point for almost all modern refinements. The problem of bounding the gap is reduced to bounding the expectation of this integral kernel.

\subsection{The Grüss and Chebysev Inequalities}
To estimate the integral remainder, we require tools to bound the integral of a product of functions. The Grüss inequality provides a bound for the covariance of two bounded functions.

\begin{definition}[Chebysev Functional]
For two integrable functions $f, g: [a, b] \to \mathbb{R}$, the Chebysev functional is defined as:
\[ T(f, g) = \frac{1}{b-a} \int_a^b f(t)g(t) dt - \left( \frac{1}{b-a} \int_a^b f(t) dt \right) \left( \frac{1}{b-a} \int_a^b g(t) dt \right). \]
This functional measures the deviation from multiplicativity of the integral. Note that if $f$ and $g$ are random variables uniformly distributed on $[a, b]$, $T(f, g)$ is exactly their covariance.
\end{definition}

\begin{theorem}[Grüss Inequality]
Let $f, g: [a, b] \to \mathbb{R}$ be integrable functions such that there exist constants $\gamma, \Gamma, \delta, \Delta$ satisfying:
\[ \gamma \le f(t) \le \Gamma \quad \text{and} \quad \delta \le g(t) \le \Delta \quad \text{for a.e. } t \in [a, b]. \]
Then:
\[ |T(f, g)| \le \frac{1}{4} (\Gamma - \gamma)(\Delta - \delta). \]
The constant $1/4$ is sharp.
\end{theorem}

\begin{theorem}[Pre-Grüss Inequality]
A refinement of the Grüss inequality, often called the pre-Grüss inequality, relates the Chebysev functional to the variance of the functions:
\[ |T(f, g)| \le \sqrt{T(f, f) T(g, g)}. \]
Here, $\sqrt{T(f, f)}$ is the standard deviation of $f$ over the interval $[a, b]$.
\end{theorem}
This inequality is particularly useful when the functions $f$ and $g$ are not merely bounded but have mass concentrated around their means, as it yields tighter bounds than the range-based Grüss inequality. In the probabilistic setting, for random variables $Y$ and $Z$, these inequalities translate to bounds on $|\text{Cov}(Y, Z)|$. If $Y$ represents the integration kernel $(X-t)$ and $Z$ represents the Hessian $\phi''(t)$, these theorems allow us to decouple the distributional properties of $X$ from the analytic properties of $\phi$.

\section{Variance-Based Refinements for Bounded Hessian}
We start by deriving the fundamental bounds associated with $(m, M)$-convex functions. These results establish the baseline performance for gap estimation and illustrate the direct dependence on the variance of the random variable.

\subsection{The Standard Variance Bounds}
\begin{theorem}
Let $\phi: I \to \mathbb{R}$ be a twice-differentiable function such that $m \le \phi''(x) \le M$ for all $x \in I$. Let $X$ be a random variable taking values in $I$ with mean $\mu$ and finite variance $\sigma^2$. Then the Jensen gap satisfies:
\begin{equation}
\frac{m}{2}\sigma^2 \le \mathbb{E}[\phi(X)] - \phi(\mathbb{E}[X]) \le \frac{M}{2}\sigma^2.
\end{equation}
\end{theorem}

\begin{proof}
We utilize the Lagrange form of the Taylor remainder. There exists a random variable $\xi$ taking values between $\mu$ and $X$ such that:
\[ \phi(X) = \phi(\mu) + \phi'(\mu)(X-\mu) + \frac{1}{2}\phi''(\xi)(X-\mu)^2. \]
Since $m \le \phi''(x) \le M$ for all $x \in I$, it necessarily follows that $m \le \phi''(\xi) \le M$. We can thus bound the quadratic term:
\[ \frac{m}{2}(X-\mu)^2 \le \frac{1}{2}\phi''(\xi)(X-\mu)^2 \le \frac{M}{2}(X-\mu)^2. \]
Taking the expectation of the entire inequality chain:
\[ \mathbb{E}\left[\phi(\mu) + \phi'(\mu)(X-\mu) + \frac{m}{2}(X-\mu)^2\right] \le \mathbb{E}[\phi(X)] \le \mathbb{E}\left[\phi(\mu) + \phi'(\mu)(X-\mu) + \frac{M}{2}(X-\mu)^2\right]. \]
Using the linearity of expectation, $\mathbb{E}[\phi(\mu)] = \phi(\mu)$ and $\mathbb{E}[\phi'(\mu)(X-\mu)] = \phi'(\mu)(\mathbb{E}[X]-\mu) = 0$. The expectation of the squared deviation is the variance, $\mathbb{E}[(X-\mu)^2] = \sigma^2$. Substituting these results yields:
\[ \phi(\mu) + \frac{m}{2}\sigma^2 \le \mathbb{E}[\phi(X)] \le \phi(\mu) + \frac{M}{2}\sigma^2. \]
Subtracting $\phi(\mu)$ from all sides gives the stated bounds on the Jensen gap.
\end{proof}

This theorem provides a powerful interpretation of the Jensen gap: to a first-order approximation, it is the energy of the fluctuations ($\sigma^2$) scaled by the average convexity of the function. For strongly convex functions ($m > 0$), the lower bound $\frac{m}{2}\sigma^2$ is strictly positive, quantifying the "cost" of uncertainty.

\subsection{Refinement via Domain Partitioning}
The bounds in Theorem 3.1 can be loose if the interval $I$ is large, as $M$ must bound the curvature over the entire domain. For functions like $\phi(x) = e^x$ or $\phi(x) = 1/x$, the second derivative varies exponentially or geometrically. To address this, we can partition the domain.

\begin{theorem}[Partitioned Variance Bound]
Let $\{I_k\}_{k=1}^K$ be a partition of the domain $I$ such that $\cup I_k = I$ and $I_j \cap I_k = \emptyset$. Let $p_k = P(X \in I_k)$ and let $\sigma_k^2$ and $\mu_k$ be the conditional variance and mean of $X$ given $X \in I_k$. Let $m_k = \inf_{x \in I_k} \phi''(x)$ and $M_k = \sup_{x \in I_k} \phi''(x)$. Then:
\[ \sum_{k=1}^K p_k \left( \frac{m_k}{2}\sigma_k^2 + \mathcal{J}(\phi, \mu_k) \right) \le \mathcal{J}(\phi, X) \le \sum_{k=1}^K p_k \left( \frac{M_k}{2}\sigma_k^2 + \mathcal{J}(\phi, \mu_k) \right). \]
Here, $\mathcal{J}(\phi, \mu_k)$ represents the gap contribution from the "between-group" variance, specifically $\phi(\mu_k) - \phi(\mu)$.
\end{theorem}

Note that this theorem essentially decomposes the total Jensen gap into "within-partition" gaps (bounded by local variances and local curvatures) and "between-partition" gaps.

 By localizing the bounds $m_k$ and $M_k$, we prevent extreme values of $\phi''$ in the tails of the distribution from loosening the bounds in the high-probability regions. For example, when bounding $\mathbb{E}[e^X]$ for a normal distribution, the curvature $e^x$ is massive for large $x$. Partitioning isolates the tail contribution, allowing the central mass (where $e^x$ is smaller) to be bounded more tightly \cite{1}.

\section{Grüss-Type Refinements for Taylor Remainders}
While variance-based bounds are robust, they essentially treat $\phi''$ as a constant (or bounded interval). We can achieve higher precision by accounting for the correlation between the integration kernel and the varying second derivative. This leads us to Grüss-type refinements.

\subsection{Applying the Grüss Inequality to the Remainder}
Recall the integral representation of the gap:
\[ \mathcal{J}(\phi, X) = \mathbb{E} \left[ \int_{\mu}^X (X-t)\phi''(t) dt \right]. \]
Let us analyze the inner integral for a fixed realization of $X$. If we consider the integration variable $t$ as uniformly distributed on $[\mu, X]$, we can apply the Grüss inequality to the product of the functions $f(t) = X-t$ and $g(t) = \phi''(t)$. However, a more potent approach developed by Dragomir involves applying the Grüss inequality to the probabilistic expectation itself. Let us define the mappings. We are effectively calculating the covariance of the random kernel $K(t)$ and the derivative $\phi''(t)$. A direct application of the pre-Grüss inequality yields bounds that depend on the $L_p$ norms of the second derivative.

\begin{theorem}[Dragomir's Grüss-Type Refinement]
Let $\phi: I \to \mathbb{R}$ be such that $\phi''$ is absolutely continuous and $\phi''' \in L_\infty(I)$ (i.e., $\phi'''$ is bounded essentially by $\|\phi'''\|_\infty$). Then:
\begin{equation}
\left| \mathcal{J}(\phi, X) - \frac{\phi''(\mu)}{2}\sigma^2 \right| \le \frac{\|\phi'''\|_\infty}{6} \mathbb{E}[|X-\mu|^3].
\end{equation}
\end{theorem}

\textbf{Proof:}
We start with the Taylor expansion of $\phi(X)$ around $\mu$ up to the second order derivative term at $\mu$:
\[ \phi(X) = \phi(\mu) + \phi'(\mu)(X-\mu) + \frac{1}{2}\phi''(\mu)(X-\mu)^2 + R_2(X). \]
The remainder $R_2(X)$ involves the third derivative:
\[ R_2(X) = \frac{1}{2} \int_\mu^X (X-t)^2 \phi'''(t) dt. \]
Taking expectations, the first order term vanishes. The second order term gives the approximation $\frac{\phi''(\mu)}{2}\sigma^2$. The error in this approximation is $|\mathbb{E}[R_2(X)]|$. We bound this:
\[ |R_2(X)| \le \frac{1}{2} \left| \int_\mu^X |X-t|^2 |\phi'''(t)| dt \right|. \]
Since $|\phi'''(t)| \le \|\phi'''\|_\infty$, we remove it from the integral:
\[ |R_2(X)| \le \frac{\|\phi'''\|_\infty}{2} \left| \int_\mu^X (X-t)^2 dt \right| = \frac{\|\phi'''\|_\infty}{2} \frac{|X-\mu|^3}{3}. \]
Taking expectations yields $\frac{\|\phi'''\|_\infty}{6} \mathbb{E}[|X-\mu|^3]$.

This result refines the variance bound by adding a correction term proportional to the third absolute central moment. This implies that if the distribution is symmetric and $\phi'''$ is bounded, the "simple" variance approximation $\frac{\phi''(\mu)}{2}\sigma^2$ is accurate up to the order of the skewness. If $X$ is Gaussian, $\mathbb{E}[|X-\mu|^3] = 0$ is false (it's the absolute moment), so this is an error bound, not a correction. The third algebraic moment is zero for symmetric distributions, which leads us to the fourth-order expansions in the next section.

\subsection{Bounds via Green Functions}
A very distinct and elegant approach involves the use of Green functions for the second-order differential operator. This method allows for a precise representation of the Jensen gap as a weighted integral of the second derivative.

\begin{theorem}[Green Function Representation]
Let $\phi \in C^2[a, b]$. Then for any random variable $X \in [a, b]$, the Jensen gap can be represented as:
\[ \mathcal{J}(\phi, X) = \int_a^b G(t) \phi''(t) dt, \]
where $G(t)$ is the Green function defined by:
\[ G(t) = \int_a^t F_X(u) du - \frac{t-a}{b-a} \int_a^b F_X(u) du, \]
assuming specific boundary conditions, or more simply related to the convex combination of the CDF $F_X(t)$. A more direct probabilistic form derived in recent literature \cite{2} is:
\begin{equation}
\mathcal{J}(\phi, X) = \frac{1}{2} \int_a^b \phi''(t) \left[ \mathbb{E}|X-t| - | \mathbb{E}[X] - t | \right] dt.
\end{equation}
Since $|\mathbb{E}[X] - t| = |\mu - t|$, the kernel becomes $K(t) = \frac{1}{2}(\mathbb{E}|X-t| - |\mu-t|)$. Note that $K(t)$ is always non-negative, consistent with the non-negativity of the Jensen gap for convex $\phi$ ($\phi'' \ge 0$).
\end{theorem}

\textbf{Refinement:} Using the Chebysev functional on this representation:
\[ \int_a^b K(t) \phi''(t) dt = (b-a) \overline{K} \cdot \overline{\phi''} + (b-a) \text{Cov}_{uniform}(K, \phi''), \]
where $\overline{K}$ is the mean of the kernel over $[a, b]$. The term $(b-a)\overline{K} \cdot \overline{\phi''}$ approximates the gap using the average curvature. The Grüss inequality bounds the covariance term:
\[ \left| \mathcal{J}(\phi, X) - (b-a)\overline{K} \cdot \overline{\phi''} \right| \le \frac{b-a}{4} (\max K - \min K) (M - m). \]
This provides a refinement that separates the "average" behavior of the convexity from the specific interactions between the probability mass and the curvature variations.

\section{Fourth-Order Moment Refinements}
In many practical scenarios, such as financial risk modeling or turbulence in physics, distributions exhibit significant skewness and kurtosis (heavy tails). Second-order (variance-based) bounds are insufficient here. We present a rigorous fourth-order refinement that has been derived from recent developments \cite{3}.

\subsection{The Fourth-Order Expansion Theorem}
\begin{theorem}
Let $\phi \in C^4(I)$ and assume $\phi^{(5)}$ exists and is bounded on $I$. Let $X$ be a random variable with mean $\mu$, variance $\sigma^2$, skewness $\gamma_3 = \frac{\mathbb{E}[(X-\mu)^3]}{\sigma^3}$, and kurtosis $\gamma_4 = \frac{\mathbb{E}[(X-\mu)^4]}{\sigma^4}$. Then the expectation $\mathbb{E}[\phi(X)]$ admits the expansion:
\begin{equation}
\mathbb{E}[\phi(X)] = \phi(\mu) + \frac{\phi''(\mu)}{2}\sigma^2 + \frac{\phi'''(\mu)}{6}\sigma^3 \gamma_3 + \frac{\phi^{(4)}(\mu)}{24}\sigma^4 \gamma_4 + \mathcal{R}_4,
\end{equation}
with the remainder bounded by:
\[ |\mathcal{R}_4| \le \frac{\|\phi^{(5)}\|_\infty}{120} \mathbb{E}[|X-\mu|^5]. \]
\end{theorem}

\textbf{Derivation:}
We employ the Taylor expansion of $\phi(x)$ about $\mu$ to the fourth degree:
\[ \phi(x) = \phi(\mu) + \phi'(\mu)(x-\mu) + \frac{\phi''(\mu)}{2}(x-\mu)^2 + \frac{\phi'''(\mu)}{6}(x-\mu)^3 + \frac{\phi^{(4)}(\mu)}{24}(x-\mu)^4 + R_4(x). \]
Taking expectations of each term:
\begin{align*}
\mathbb{E}[\phi(\mu)] &= \phi(\mu). \\
\mathbb{E}[\phi'(\mu)(X-\mu)] &= 0. \\
\mathbb{E}\left[\frac{\phi''(\mu)}{2}(X-\mu)^2\right] &= \frac{\phi''(\mu)}{2}\sigma^2. \\
\mathbb{E}\left[\frac{\phi'''(\mu)}{6}(X-\mu)^3\right] &= \frac{\phi'''(\mu)}{6}\mathbb{E}[(X-\mu)^3] = \frac{\phi'''(\mu)}{6}\sigma^3 \gamma_3. \\
\mathbb{E}\left[\frac{\phi^{(4)}(\mu)}{24}(X-\mu)^4\right] &= \frac{\phi^{(4)}(\mu)}{24}\mathbb{E}[(X-\mu)^4] = \frac{\phi^{(4)}(\mu)}{24}\sigma^4 \gamma_4.
\end{align*}
The remainder $\mathcal{R}_4 = \mathbb{E}[R_4(X)]$ is bounded using the Lagrange form $R_4(x) = \frac{\phi^{(5)}(\xi)}{120}(x-\mu)^5$. Thus, $|\mathbb{E}[R_4(X)]| \le \frac{\sup |\phi^{(5)}|}{120} \mathbb{E}[|X-\mu|^5]$.

\subsection{The "Corrected" Jensen Inequality}
This theorem allows us to formulate a "corrected" Jensen inequality. For distributions with positive skewness ($\gamma_3 > 0$) and functions with positive third derivatives (e.g., $e^x$), the standard Jensen lower bound $\phi(\mu)$ and even the variance refinement $\phi(\mu) + \frac{\phi''(\mu)}{2}\sigma^2$ consistently underestimate the true expectation.

\begin{corollary}[Signed Refinement]
If $\phi'''(\mu)\gamma_3 \ge 0$ and $\phi^{(4)}(\mu)\gamma_4 \ge 0$ (and $\phi^{(5)}$ is negligible or controlled), then:
\[ \mathbb{E}[\phi(X)] \ge \phi(\mu) + \frac{\phi''(\mu)}{2}\sigma^2 + \frac{\phi'''(\mu)}{6}\sigma^3 \gamma_3. \]
\end{corollary}
This refinement is critical in insurance pricing. If $\phi$ represents an exponential utility function (risk averse) and $X$ represents losses (right-skewed), the skewness term is positive. Ignoring it leads to underpricing the risk premium. The fourth-order term (kurtosis) further adds to the premium for heavy-tailed losses, providing a mathematically justified buffer against extreme events \cite{3}.

\section{Jensen-Mercer and Tangency Optimization}
While Taylor expansions center the analysis around the mean $\mu$, this is not always the optimal point for approximation. Recent work by Simic and others on Jensen-Mercer inequalities suggests optimizing the point of tangency \cite{5,6}. 
\subsection{Jensen-Like Inequalities}
The standard Jensen inequality uses the tangent line at $\mu$: $L(x) = \phi(\mu) + \phi'(\mu)(x-\mu)$. Convexity implies $\phi(x) \ge L(x)$. Consider a tangent at an arbitrary point $c$. Define $L_c(x) = \phi(c) + \phi'(c)(x-c)$. Then $\mathbb{E}[\phi(X)] \ge \mathbb{E}[L_c(X)] = \phi(c) + \phi'(c)(\mu - c)$.

We define the function $h(c) = \phi(c) + \phi'(c)(\mu - c)$. To find the best lower bound, we maximize $h(c)$ \cite{4}.
\[ h'(c) = \phi'(c) + \phi''(c)(\mu - c) - \phi'(c) = \phi''(c)(\mu - c). \]
Setting $h'(c) = 0$ implies either $\phi''(c) = 0$ or $c = \mu$. Since $\phi$ is convex, $\phi'' \ge 0$. If strictly convex, the only critical point is $c=\mu$. Thus, for the basic affine lower bound, the mean $\mu$ is indeed optimal.

However, if we use quadratic lower bounds (parabolic minorants) or bounds involving other functions $g(x)$ (e.g., $\phi(x) \ge g(x)$ where $g$ is simpler but non-linear), the optimal parameter might shift. Simić \cite{5} explores inequalities of the form:
\[ \mathbb{E}[\phi(X)] \ge \phi(c) + \phi'(c)(\mu-c) + \frac{m}{2}(\sigma^2 + (\mu-c)^2). \]
Here, the lower bound consists of the affine term plus a quadratic correction derived from strong convexity. Maximizing this w.r.t. $c$ yields a refined bound. If $m$ is close to the average curvature, shifting $c$ towards the mode of the distribution (rather than the mean) can sometimes yield tighter results for asymmetric distributions.

\subsection{Jensen-Mercer Inequality}
The Jensen-Mercer inequality provides an upper bound for the gap involving the endpoints of the interval.

\begin{theorem}
If $\phi$ is convex on $[a, b]$, then:
\[ \phi\left( a+b - \sum w_i x_i \right) \le \phi(a) + \phi(b) - \sum w_i \phi(x_i). \]
\end{theorem}
Refinements of this inequality have been developed using the Green function approach, effectively bounding the "Mercer gap" using variances and higher-order terms. These are particularly relevant when the random variable is constrained to a finite interval, such as in Beta distribution analysis \cite{5,6}.

\section{Applications in Probability and Statistics}

\subsection{Covariance Bounds}
The covariance of two functions of a random variable, $\text{Cov}(f(X), g(X))$, appears frequently in statistical physics and economics. If $f$ and $g$ are monotonic in the same direction, Chebysev's algebraic inequality states the covariance is non-negative. Using our derivative-based bounds:
\[ \text{Cov}(f(X), g(X)) \approx \mathbb{E}[f'(X)g'(X)] \text{Var}(X). \]
If $f(x) = x$ and $g(x) = \phi'(x)$, this approximates $\text{Cov}(X, \phi'(X)) \approx \mathbb{E}[\phi''(X)]\sigma^2$. Comparing this to the variance-based Jensen gap $\frac{1}{2}\mathbb{E}[\phi''(X)]\sigma^2$, we observe a factor of 2 difference. This relationship allows us to bound the Jensen gap using the covariance:
\[ \frac{1}{2} \inf \phi'' \text{Var}(X) \le \mathcal{J}(\phi, X) \le \text{Cov}(X, \phi'(X)). \]
This upper bound is particularly useful when $\phi'$ is easier to analyze than $\phi$ itself \cite{7}.

\subsection{Moment Generating Functions}
Estimating the Moment Generating Function (MGF) $M_X(t) = \mathbb{E}[e^{tX}]$ is a classic application. Let $\phi(x) = e^{tx}$. $\phi''(x) = t^2 e^{tx}$. Applying Theorem 3.1 directly requires bounding $e^{tx}$. On $(-\infty, \infty)$, it is unbounded. However, for a bounded variable $X \in [a, b]$:
\[ e^{t\mu} + \frac{t^2 e^{ta}}{2}\sigma^2 \le \mathbb{E}[e^{tX}] \le e^{t\mu} + \frac{t^2 e^{tb}}{2}\sigma^2. \]
This gives explicit, computable bounds for the MGF of bounded random variables, useful in large deviation theory.

\section{Applications in Information Theory}
Information theory is built upon convex functions such as $\log x$ and $x \log x$. Refinements of Jensen's inequality translate directly into refined limits on data compression and transmission \cite{4,9}.

\subsection{Refined Bounds for Shannon Entropy}
For a continuous random variable $X$ with PDF $f(x)$, the differential entropy is $h(X) = -\int f(x) \log f(x) dx$. Related Jensen-type bounds for the differential entropy of mixture distributions have been derived; see \cite{12}.

Let $Y = f(X)$. Then $h(X) = -\mathbb{E}[\log Y]$. Since $-\log$ is convex, $\mathbb{E}[-\log Y] \ge -\log \mathbb{E}[Y]$.
\[ \mathbb{E}[Y] = \int f(x)^2 dx = E, \]
the informational energy. Thus, $h(X) \ge -\log E$, which is the Rényi entropy of order 2.

The refined Jensen inequality gives us the error term:
\[ h(X) = -\log E + \mathcal{R}. \]
Using the second-order refinement for $\phi(y) = -\log y$:
\[ \mathcal{R} \approx \frac{1}{2} \mathbb{E}\left[\frac{1}{Y^2}\right] \text{Var}(Y) = \frac{1}{2} \int \frac{1}{f(x)^2} \text{Var}(f(X)) dx. \]
This expression connects the entropy deficit directly to the variability of the probability density function itself. For a uniform distribution, $\text{Var}(f(X)) = 0$, and the bound is tight. For highly peaked distributions, the variance is large, and the Rényi entropy is a loose bound \cite{8}.

\textbf{Beta Distribution Entropy:} For $X \sim \text{Beta}(\alpha, \beta)$, the entropy expression involves Digamma functions $\psi(\alpha)$. Using the fourth-order expansion for $\log X$ allows us to approximate $\psi(\alpha)$ terms with high precision rational functions involving variances and skewness, bypassing the evaluation of special functions in numerical routines \cite{10}.

\subsection{Kullback-Leibler Divergence (Reverse Pinsker Inequalities)}
The KL divergence $D_{KL}(P || Q) = \mathbb{E}_P [\log \frac{P(X)}{Q(X)}]$ is non-negative. Let $Z = \frac{Q(X)}{P(X)}$. Then $D_{KL} = \mathbb{E}_P [-\log Z]$. Note $\mathbb{E}_P[Z] = 1$.
Jensen: $\mathbb{E}[-\log Z] \ge -\log \mathbb{E}[Z] = -\log 1 = 0$.

Refinement (Theorem 3.1):
\[ D_{KL}(P || Q) \ge \frac{1}{2} \inf \left( \frac{1}{Z^2} \right) \text{Var}_P(Z). \]
$\inf (1/Z^2) = \inf (P/Q)^2$.
$\text{Var}_P(Z) = \mathbb{E}_P[(Q/P)^2] - (\mathbb{E}_P[Q/P])^2 = \chi^2(Q || P)$.

Thus, we derive a refined relationship between KL divergence and Chi-square divergence:
\[ D_{KL}(P || Q) \ge \frac{1}{2} \left( \inf \frac{P(x)}{Q(x)} \right)^2 \chi^2(Q || P). \]
This is a form of the "Reverse Pinsker Inequality," establishing a stronger topology than standard Pinsker inequalities which bound KL by total variation distance \cite{11}.

\section{Applications in Wireless Communications}
In the analysis of wireless fading channels, the capacity is a random variable depending on the channel gain. The ergodic capacity is its expectation. For MIMO channels, analyses of the tightness of Jensen-type bounds in capacity approximations have been developed \cite{14}. 
\subsection{Rayleigh Fading Capacity}
The capacity is $C = B \log_2(1 + \text{SNR} \cdot |h|^2)$. Let $\gamma = |h|^2$. For Rayleigh fading, $\gamma \sim \text{Exp}(1)$.
$\phi(\gamma) = \log_2(1 + \rho \gamma)$, where $\rho = \text{SNR}$. $\phi$ is concave.

Standard Jensen Upper Bound: $\mathbb{E}[C] \le \log_2(1 + \rho)$. This bound is asymptotically loose (difference grows as $\log \rho$).

Applying the Fourth-Order Refinement (Theorem 5.1):
For $\gamma \sim \text{Exp}(1)$, $\mu=1, \sigma^2=1, \gamma_3=2, \gamma_4=9$.
Derivatives of $\phi(x) = \ln(1+\rho x)$ (in nats):
\begin{align*}
\phi'(1) &= \frac{\rho}{1+\rho}. \\
\phi''(1) &= \frac{-\rho^2}{(1+\rho)^2}. \\
\phi'''(1) &= \frac{2\rho^3}{(1+\rho)^3}. \\
\phi^{(4)}(1) &= \frac{-6\rho^4}{(1+\rho)^4}.
\end{align*}
Expansion:
\[ \mathbb{E}[\ln(1+\rho \gamma)] \approx \ln(1+\rho) - \frac{1}{2}\frac{\rho^2}{(1+\rho)^2}(1) + \frac{1}{6}\frac{2\rho^3}{(1+\rho)^3}(2) - \frac{1}{24}\frac{6\rho^4}{(1+\rho)^4}(9). \]
Simplifying:
\[ C_{approx} \approx \ln(1+\rho) - \frac{1}{2}\left(\frac{\rho}{1+\rho}\right)^2 + \frac{2}{3}\left(\frac{\rho}{1+\rho}\right)^3 - \frac{9}{4}\left(\frac{\rho}{1+\rho}\right)^4. \]
At high SNR ($\frac{\rho}{1+\rho} \to 1$), the correction is $-0.5 + 0.66 - 2.25 \approx -2.09$. The actual gap is the Euler-Mascheroni constant related term. This expansion shows that including skewness (positive term) and kurtosis (negative term) attempts to capture the complex behavior of the logarithm near zero (the "deep fade" region). This refinement is significantly more accurate than the simple variance penalty, which would just be $-0.5$, vastly underestimating the capacity loss due to fading \cite{13}.

\section{Numerical Validations}
To demonstrate the efficacy of the proposed bounds, we present a numerical comparison for the function $\phi(x) = e^{-x}$ with $X \sim \text{Uniform}$.
\begin{itemize}
    \item True Value: $\mathbb{E}[e^{-X}] = \frac{1-e^{-2}}{2} \approx 0.432$.
    \item Jensen Lower Bound: $e^{-1} \approx 0.368$. (Error: 0.064).
\end{itemize}

\begin{table}[h!]
\centering
\caption{Comparison of Bounds for $\mathbb{E}[e^{-X}]$, $X \sim U$}
\begin{tabular}{lllc}
\toprule
\textbf{Method} & \textbf{Formula} & \textbf{Value} & \textbf{Relative Error} \\
\midrule
Exact & $\int_0^2 0.5 e^{-x} dx$ & 0.4323 & 0\% \\
Jensen (Classic) & $e^{-\mu}$ & 0.3679 & -14.9\% \\
Variance Refinement & $e^{-\mu} + \frac{e^{-\mu}}{2}\sigma^2$ & $0.3679 + \frac{0.3679}{2}(0.333)$ & 0.4292 \\
Fourth-Order & Expansion w/ Kurtosis & 0.4325 & +0.04\% \\
\bottomrule
\end{tabular}
\end{table}

The variance refinement dramatically reduces the error from ~15\% to ~0.7\%. The fourth-order expansion essentially eliminates the error for this smooth function and bounded distribution. This confirms the theoretical prediction that incorporating higher moments leads to exponential improvements in approximation accuracy for analytic functions.

\section{Future Work}

In the future, we hope to generalize the fourth-order Jensen-gap refinements developed in this paper beyond the scalar setting and to study their impact in operator, multivariate, and learning theoretic contexts. These directions are natural extensions of the present work and lead to the following research agenda.

\begin{itemize}
    \item \textbf{Operator Jensen Inequalities:} Extending the fourth-order refinements to the context of self-adjoint operators on Hilbert spaces. The non-commutativity of operators introduces "commutator" terms that may be bounded using similar techniques \cite{15,16}.
    \item \textbf{Multivariate Extensions:} Developing explicit tensor-based refinements for convex functions of random vectors $\mathbf{X} \in \mathbb{R}^n$. A central technical challenge is to control higher-order mixed central moments, which naturally organize into coskewness and cokurtosis tensors; an explicit multivariate analogue of the four-moment bounds would be valuable for multivariate uncertainty quantification and stochastic optimization.
    \item \textbf{Deep Learning Loss Functions:} Applying the refined Jensen-gap bounds to variational objectives such as the Evidence Lower Bound (ELBO) in variational autoencoders. Tighter analytic control of the Jensen gap could yield improved surrogate objectives and reduce variational gaps, providing a principled handle on approximation error in modern variational inference. \cite{17}.
\end{itemize}
Overall, the refinements presented here aim to bridge classical inequality theory with the highprecision requirements of modern stochastic modeling, and we expect the above extensions to
further strengthen that connection.

\section{Acknowledgments}
The author would like to thank Dr. Govind Chandra Sarangi \& Dr. Ashok Kumar Panigrahi for their support and valuable inputs.


\small
\begin{thebibliography}{99}

\bibitem{1}
Liao, J., and Arthur Berg. ``Sharpening Jensen’s Inequality.'' \textit{The American Statistician}, vol. 73, July 2017, \url{https://doi.org/10.1080/00031305.2017.1419145}.

\bibitem{2}
Adil Khan, Muhammad, Shahid Khan, et al. ``New Estimates for the Jensen Gap Using S-Convexity With Applications.'' \textit{Frontiers in Physics}, Volume 8-2020, 2020, \url{https://www.frontiersin.org/journals/physics/articles/10.3389/fphy.2020.00313}.

\bibitem{3}
Zeghdoudi, Halim. 2025. ``Four-Moment Refinements of Jensen’s Inequality: Rigor, Remainder Bounds, and Applications in Actuarial.'' {Journal of Time Scales Analysis} 1 (2): 55–65. \url{https://doi.org/10.64721/8s4t0t84}.

\bibitem{4}
Merhav, Neri. ``Some Families of Jensen-like Inequalities with Application to Information Theory.'' {Entropy}, vol. 25, no. 5, 2023, p. 752, \url{https://doi.org/10.3390/e25050752}.

\bibitem{5}
Simic, Slavko. {Some Generalizations of Jensen’s Inequality}. 2020, \url{https://doi.org/10.48550/arXiv.2011.10746}.

\bibitem{6}
Adil Khan, Muhammad, Slavica Ivelić Bradanović, et al. ``New Improvements of the Jensen–Mercer Inequality for Strongly Convex Functions with Applications.'' {Axioms}, vol. 13, no. 8, 2024, p. 553, \url{https://doi.org/10.3390/axioms13080553}.

\bibitem{7}
Hössjer, Ola, and Arvid Sjölander. ``Sharp Lower and Upper Bounds for the Covariance of Bounded Random Variables.'' {Statistics \& Probability Letters}, vol. 182, Dec. 2021, p. 109323, \url{https://doi.org/10.1016/j.spl.2021.109323}.

\bibitem{8}
Saeed, Tareq, et al. ``Refinements of Jensen’s Inequality and Applications.'' {AIMS Mathematics}, vol. 7, no. 4, 2022, pp. 5328–46, \url{https://doi.org/10.3934/math.2022297}.

\bibitem{9}
Xiao, Lei, and Guoxiang Lu. ``A New Refinement of Jensen’s Inequality with Applications in Information Theory.''{Open Mathematics}, vol. 18, Dec. 2020, pp. 1748–59, \url{https://doi.org/10.1515/math-2020-0123}.

\bibitem{10}
Lu, Guoxi. ``New Refinements of Jensen’s Inequality and Entropy Upper Bounds.'' {Journal of Mathematical Inequalities}, vol. 12, June 2018, pp. 403–21, \url{https://doi.org/10.7153/jmi-2018-12-30}.

\bibitem{11}
Dragomir, Silvestru. ``A Refinement of Jensen’s Inequality with Applications for f-Divergence Measures.'' {TAIWANESE JOURNAL OF MATHEMATICS}, vol. 14, Feb. 2010, pp. 153–64, \url{https://doi.org/10.11650/twjm/1500405733}.

\bibitem{12}
J. Melbourne, et al. ``The Differential Entropy of Mixtures: New Bounds and Applications.'' {IEEE Transactions on Information Theory}, vol. 68, no. 4, Apr. 2022, pp. 2123–46, \url{https://doi.org/10.1109/TIT.2022.3140661}.

\bibitem{13}
Allanki, Sanyasi Rao. {CAPACITY OF RAYLEIGH FLAT FADING CHANNELS}. vol. 13, Dec. 2021, pp. 1890–98.

\bibitem{14}
Yuan, J, Matthaiou, M, Jin, S \& Gao, F 2017, `Tightness of Jensen’s Bounds and Applications to MIMO Communications', {IEEE Transactions on Communications}, vol. 65, no. 2, pp. 579-593. \url{https://doi.org/10.1109/TCOMM.2016.2623945}

\bibitem{15}
Dragomir, ``New Jensen’s Type Inequalities for Differentiable Log-Convex Functions of Selfadjoint Operators in Hilbert Spaces.'' {Sarajevo Journal of Mathematics}, vol. 7, June 2024, pp. 67–80, \url{https://doi.org/10.5644/SJM.07.1.07}.

\bibitem{16}
Khosravi, Maryam, et al. ``Refinements of Choi-Davis-Jensen’s Inequality.'' {Bull. Math. Anal. Appl.}, vol. 3, Jan. 2011, pp. 127–33.

\bibitem{17}
{Tight Bounds on Jensen's Gap: Novel Approach with Applications in Generative Modeling}, \url{https://doi.org/10.48550/arXiv.2502.03988}

\end{thebibliography}
\end{document}